\documentclass[journal]{IEEEtran}
\ifCLASSINFOpdf
\else
\fi
\usepackage{verbatim}
\usepackage{cite}
\usepackage{amsmath, amsthm}
\usepackage{amssymb}
\usepackage{mathrsfs}
\usepackage{graphicx}
\usepackage{subfigure}
\usepackage{epstopdf}
\usepackage{eucal}
\usepackage{caption}
\usepackage[noend]{algpseudocode}
\usepackage{algorithmicx,algorithm}
\usepackage{float}
\usepackage{color, soul}

\theoremstyle{plain}
\newtheorem{thm}{Theorem}
\newtheorem{lem}{Lemma}
\newtheorem{ass}{Assumption}

\theoremstyle{definition}
\newtheorem{defn}{Definition}

\theoremstyle{remark}
\newtheorem{rem}{Remark}

\newtheoremstyle{noparens}%
  {}{}%
  {\itshape}{}%
  {\bfseries}{.}%
  { }%
  {\thmname{#1}\thmnumber{ #2}\mdseries\thmnote{ #3}}
\theoremstyle{noparens}
\newcommand{\norm}[1]{\left\lVert#1\right\rVert}

\allowdisplaybreaks[4]
\begin{document}
\title{
Model-free optimal control of discrete-time systems with additive and multiplicative noises}
\author{Jing Lai, ~\IEEEmembership{student,}
        Junlin Xiong,~\IEEEmembership{Member~IEEE,}
        Zhan Shu,~\IEEEmembership{Senior Member~IEEE,}
\thanks{This work was financially supported by the National Natural Science Foundation of China under Grant 61374026 and Grant 61773357.}
\thanks{J. Lai and J. Xiong are with the Department of Automation, University of Science and
Technology of China, Hefei 230026, China (e-mail:lj120@mail.ustc.edu.cn; junlin.xiong@gmail.com).}
\thanks{Z. Shu is with the Department of Electrical and Computer Engineering, University of Alberta, Edmonton, AB T6G 2R3, Canada (e-mail:
hustd8@gmail.com).}
}
\markboth{---}%
{Shell \MakeLowercase{\textit{et al.}}: Bare Demo of IEEEtran.cls for IEEE Journals}

\maketitle
\begin{abstract}
This paper investigates the optimal control problem for a class of discrete-time
 stochastic systems subject to additive and multiplicative noises.
A stochastic Lyapunov equation and a stochastic algebra Riccati
equation are established for the existence of the optimal admissible control policy.
A model-free reinforcement learning algorithm is proposed to learn the optimal admissible control policy using the data of the system states and inputs without requiring any knowledge of the system matrices.
It is proven that the learning algorithm converges to the optimal admissible control policy.
The implementation
of the model-free algorithm is based on batch least squares and numerical average.
The proposed algorithm is illustrated through a numerical example, which shows our algorithm outperforms other policy iteration algorithms.
\end{abstract}
\IEEEpeerreviewmaketitle

\section{Introduction}
\IEEEPARstart{R}{einforcement} learning (RL) \cite{bertsekas2019reinforcement}
has been widely used for solving optimization problems in poorly structured
or initially unknown environments.
In the control system society, RL has been extensively studied to solve
the optimal control problem without
 requiring any knowledge of the system matrices \cite{bradtke1994adaptive,DV2009,kiumarsi2017optimal}.
In particular, 
policy iteration (PI) algorithms \cite{bertsekas1996neuro} were developed to solve the problem of optimal control for deterministic systems in \cite{KIUMARSI20141167,rizvi2018output,LEE20122850,modares2014linear}.
 Recently, extensions to RL-based control design for stochastic systems \cite{XU20121017,JZP2016noise,LEONG2020108759} have emerged as well.


Stochastic LQR (linear quadratic regulator) has been widely studied based on the RL method.
 For stochastic systems with additive noises, the authors of \cite{NIPS2019_9058} developed an approximate PI algorithm to solve the stochastic LQR problem in an model-free manner.
 In \cite{pmlr-v89-abbasi-yadkori19a}, a model-free learning algorithm was presented for solving the optimal control problem where policies were updated with respect to the average of all previous Q-function estimates.
 For stochastic systems with multiplicative noises, the authors of \cite{WANG2016379} presented an value iteration learning algorithm to
find the optimal control gain, where a
model neural network was used to assist the algorithm implementation.
In \cite{WANG20181}, the stochastic optimal control problem was converted into a deterministic one
and Q-learning algorithm was adopted to solve the problem where the system matrices were required partially.
In practice, many systems suffer from both multiplicative and additive noises, see
\cite{el1999h,7820092,1017567}.
For such stochastic systems,
a model-free learning algorithm based on Ito's lemma was developed for continuous-time systems in \cite{7260102}.
However, the stochastic LQR problems under additive and multiplicative noises are far from solved.

This paper aims to design an optimal
control policy for discrete-time stochastic systems using the RL method.
The systems under consideration are subject to both multiplicative and additive noises.
All of the system matrices are completely unknown.
The infinite horizon cost function may become
infinity due to the presence of the additive noise \cite{7260102}.
Hence, an infinite horizon cost function is to be minimized
where a discount factor is employed to guarantee
the cost boundedness.
In order to develop the learning algorithm, a stochastic Lyapunov equation (SLE)
and a stochastic algebra Riccati equation (SARE) are established.
Firstly, under the assumption that the system matrices are known,
an offline PI algorithm is proposed to solve the SARE iteratively.
We prove the convergence of the offline PI algorithm.
Secondly, after introducing Q-function, an online model-free RL algorithm is proposed without knowledge of the system matrices.
By showing the above algorithms are equivalent to each other, we conclude that the online model-free RL algorithm is also convergent to the optimal control policy.
Thirdly, to implement the online model-free RL algorithm,
a numerical averages is employed to approximate the expectation and
 batch least squares (BLS) is used to obtain the iterative kernel matrix of Q-function.
Our algorithm is implemented without the assumption that the noises are measurable, which is required in the cases of  continuous-time systems with additive and multiplicative noises  \cite{JZP2016noise},\cite{ZHANG20201}.
Finally, a numerical example is presented to illustrate the obtained results.
Empirically, our algorithm outperforms other PI algorithms, but
performs slightly worse than a model-based algorithm.

\textit{Notation:}
Let $\mathbb{R}^{n\times m}$ be the set of $n\times m$ real matrices.
Let $I$ denote an identity matrix with appropriate dimensions.
Notation $\mathcal{S}^{n}$ and $\mathcal{S}_+^{n}$ denote  the
set of symmetric positive definite real matrix and the
set of symmetric positive semidefinite real matrix, respectively, with dimensions $n\times n$.
Notation
$X>Y$, where $X$ and $Y$ are real symmetric matrices, means that the matrix
$X-Y$ is positive definite.
The superscript ``$\top$'' denotes the transpose for
vectors or matrices.
Let $\rho(.)$ be the spectral radius of matrices.
The trace of a square matrix $A$ is denoted by ${\rm tr}(A)$.
We use $\norm{\cdot}$ to denote the Euclidean norm for vectors.
Let $\otimes$ denote the Kronecker product.
${\rm E}$ denotes the mathematical expectation.
For symmetric matrix $X\in \mathbb{R}^{n\times n}$,
${\rm {vech}}(X)\in \mathbb{R}^\frac{n(n+1)}{2}$ denotes the vector whose elements are the $n$ diagonal entries of $X$
and the $\frac{n(n+1)}{2}-n$ distinct entry $[X]_{ij}$;
${\rm {vecs}}(X)\in \mathbb{R}^\frac{n(n+1)}{2}$ denotes
the vector whose elements are the $n$ diagonal entries of $X$ and
the $\frac{n(n+1)}{2}-n$ distinct sums $[X]_{ij}+[X]_{ji}$.

\section{problem description}
Consider the following linear discrete-time system
\begin{equation}  \label{eq1}
x_{k+1}=Ax_k+Bu_k+(Cx_k+Du_k)d_k+w_k
\end{equation}
where $x_k\in \mathbb{R}^{n}$ is the system state at time $k$,
$u_k\in \mathbb{R}^{m}$ is the control input,
$x_0$ is the system initial state following a Gaussian distribution with zero mean and covariance $X_0$.
The matrices $A$, $C\in \mathbb{R}^{n\times n}$, $B$, $D\in \mathbb{R}^{n\times m}$ are the system matrices. 
 $d_k\in \mathbb{R}$ is the system multiplicative noise,
 $w_k\in \mathbb{R}^{n}$ is the system additive noise.
The system noise sequence $\{(d_k, w_k): k=0,1,2,\ldots\}$
is defined on a given complete probability space $(\Omega, \mathcal{F}, \mathcal{P})$.
For convenience, it is further assumed that
\begin{enumerate}
        \item $d_k$ is scalar Gaussian random variable with zero mean and covariance 1;
        \item $w_k$ is Gaussian random vector with zero mean and covariance $W\in \mathcal{S}_+^{n}$;
        \item ${\rm E}(x_0d_i)=0$, ${\rm E}(x_0w_i)=0$, ${\rm E}(d_iw_j)=0$  $\forall i,j$.
\end{enumerate}


\begin{defn}\cite{kubrusly1985mean}\label{def1}
System \eqref{eq1} with control input $u_k\equiv0$ is called asymptotically square stationary (ASS) if there exists $X\in \mathcal{S}_+^{n}$ such that $\|\underset{k\rightarrow\infty}{\lim}{\rm E}(x_kx_k^{\top})-X\|=0$ independent of the covariance matrix $X_0$.
\end{defn}

\begin{lem}\cite{kubrusly1985mean}\label{lem1}
System \eqref{eq1} with control input $u_k\equiv 0$ is ASS
if and only if
\begin{equation}\label{eq_lem1}
\rho(A\otimes A+C\otimes C)<1.
\end{equation}

\end{lem}

For system \eqref{eq1}, as stated in \cite{1099867}, the optimal control policy is linear, where the optimal control gain to be designed.

\begin{defn}\label{def2}
A control policy is called admissible for system \eqref{eq1} if the system with the control policy is ASS.
\end{defn}

\begin{rem}\label{rem_Exx}
For system \eqref{eq1} under an admissible control policy $u=Lx$,
 one has
\begin{align}
\nonumber {\rm E}(x_{k+1}x_{k+1}^{\top})&=(A+BL){\rm E}(x_kx_k^{\top})(A+BL)^{\top}\\
\nonumber&\quad+(C+DL){\rm E}(x_kx_k^{\top})(C+DL)^{\top}+W.
\end{align}
Obviously, ${\rm E}(x_kx_k^{\top})$ is positive definite due to the positive definiteness of $W$.
\end{rem}


\begin{lem} \cite{kubrusly1985mean} \label{lem2}
A control policy $u=Lx$ is admissible if and only if the following algebraic equation has a unique solution $P\in \mathcal{S}_{+}^n$ for any given $F\in \mathcal{S}_{+}^n$:
\begin{equation}\label{eq_lem2_1}
P=(A+BL)^{\top}P(A+BL)+(C+DL)^{\top}P(C+DL)+F.
\end{equation}
\end{lem}

The following assumption is essential throughout this paper.

\begin{ass}\label{ass1}
There exist linear admissible control policies for system \eqref{eq1}.
\end{ass}

For an admissible control policy $u=Lx$, define the cost function as
\begin{equation}\label{eq2}
V(x_k)={\rm E}
\big( \sum_{i=k}^{\infty}\gamma^{i-k}c(x_i,u_i) \big),
\end{equation}
where $c(x_i,u_i)\geq 0$ is called one step cost at time $i$ and $1>\gamma\geq0$ is a discount factor.
Usually, the one step cost is given by $c(x_i,u_i)=x_i^{\top}Qx_i+u_i^{\top}Ru_i$ with $Q\in \mathcal{S}^n$ and $R\in \mathcal{S}_{+}^m$.

Define $U_{{\rm ad}}$ as the set containing all the admissible control policies for system \eqref{eq1}.
The stochastic LQR problem considered in this paper is to find an optimal admissible control policy in the sense of minimizing the cost function $V(x_0)$.
 The optimal cost function is given by
\begin{equation}
\nonumber V^{\ast}(x_k)=\min_{u\in U_{\rm ad}}V(x_k).
\end{equation}

\begin{defn}\label{def4}
\cite{WANG2016379} The stochastic LQR problem is called well-posed if the optimal cost function satisfies $-\infty\leq V^{\ast}(x_k)\leq\infty$.
\end{defn}


\begin{lem}\label{lem3}
If the control policy $u=Lx$ is admissible, then the stochastic LQR problem
is well-posed and the corresponding cost function is
\begin{equation}\label{eq4}
V(x_k)={\rm E}(x_k^{\top}Px_k)+\frac{\gamma}{1-\gamma}
{\rm {tr}}(PW),
\end{equation}
where $P\in \mathcal{S}_{+}^n$ is the unique solution to the stochastic
Lyapunov equation (SLE)
\begin{align}\label{eq5}
\nonumber P&=\gamma(A+BL)^{\top}P(A+BL)+\gamma (C+DL)^{\top}P\\
&\quad \times (C+DL)+L^{\top}RL+Q.
\end{align}

\end{lem}

\begin{proof}
Substituting $u=Lx$ into system \eqref{eq1} leads to
\begin{align}
\nonumber {\rm E}& (x_{k+1}^{\top}Px_{k+1}\mid x_k)\\
\nonumber &={\rm E}\Big( \big(Ax_k+BLx_k+(Cx_k+DLx_k)d_k+w_k\big)^{\top}P\\
\nonumber &\quad \times\big(Ax_k+BLx_k+(Cx_k+DLx_k)d_k+w_k\big) \mid x_k\Big)\\
\nonumber &=x_k^{\top}\big((A+BL)^{\top}P(A+BL)+(C+DL)^{\top}P\\
\nonumber &\quad \times(C+DL)\big)x_k+{\rm {tr}}(PW).
\end{align}

Noting that $P$ satisfies SLE \eqref{eq5}, one has
\begin{align}
\nonumber {\rm E}&  (\gamma x_{k+1}^{\top}Px_{k+1}-x_k^{\top}Px_k\mid x_k)\\
\nonumber &=x_k^{\top}\big(\gamma(A+BL)^{\top}P(A+BL)+\gamma (C+DL)^{\top}P\\
\nonumber &\quad \times(C+DL)-P\big)x_k+\gamma {\rm {tr}}(PW)\\
\nonumber  &=-x_k^{\top}(Q+L^{\top}RL)x_k+\gamma {\rm {tr}}(PW),
\end{align}

which means that
\begin{align}
\nonumber x_k&^{\top}(Q+L^{\top}RL)x_k\\
\nonumber   &={\rm E}\big(-\gamma x_{k+1}^{\top}Px_{k+1}+x_k^{\top}Px_k+\gamma {\rm {tr}}(PW)\mid x_k\big).
\end{align}


Therefore,
\begin{align}\label{lem3_1}
\nonumber {\rm E}
&\big(\sum_{i=k}^{\infty}\gamma^{i-k}x_i^{\top}
(Q+L^{\top}RL)x_i\mid x_k\big)\\
\nonumber &={\rm E}\big(\sum_{i=k}^{\infty}{\rm E}
(- \gamma^{i+1-k}x_{i+1}^{\top}Px_{i+1}                                  + \gamma^{i-k}x_{i}^{\top}Px_{i}\\
\nonumber &\quad+ \gamma^{i+1-k}{\rm {tr}}(PW)\mid x_i     ) \mid x_k\big)\\
\nonumber &=x_k^{\top}Px_k-\lim_{i\rightarrow \infty}{\rm E}( \gamma^i x_{k+i}^{\top}Px_{k+i}\mid x_k) + \frac{\gamma}{1-\gamma}{\rm {tr}}(PW)\\
 &=x_k^{\top}Px_k+ \frac{\gamma}{1-\gamma}{\rm {tr}}(PW),
\end{align}
where $\underset{i\rightarrow \infty}{\lim}{\rm E}( \gamma^i x_{k+i}^{\top}Px_{k+i}\mid x_k)=0$
 due to the admissibility of the control policy.
Combining \eqref{lem3_1} and \eqref{eq2},  one can obtain equation \eqref{eq4}.
Because the optimal cost function satisfies $V^{\ast}(x_k)\leq V(x_k)$,
one has $V^{\ast}(x_k)<+\infty$.
Moreover, in view of \eqref{eq2}, it is obvious that $V^{\ast}(x_k)\geq 0$.
Therefore, the stochastic LQR problem is well-posed.
The proof is completed.
\end{proof}

\begin{rem}\label{rem3_1}
Due to the existence of additive noise, the cost function in \eqref{eq4} contains the term $\frac{\gamma}{1-\gamma}
{\rm {tr}}(PW)$ which is independent of $x_k$.
When system \eqref{eq1} does not suffer from additive noise ($w_k\equiv 0$) and $\gamma =1$,
 Lemma \ref{lem3} reduces to \cite[Lemma 1]{WANG2016379}.
 Compared with \cite{WANG2016379},
an proper discount factor is used to guarantee
the  boundedness of the cost function in our case.
\end{rem}

\begin{rem}\label{rem3}
Note that the existence of a unique positive definite solution to SLE \eqref{eq5}
cannot guarantee the admissibility of $L$.
Define $M=(A+BL)\otimes (A+BL)+(C+DL)\otimes (C+DL)$.
For a control policy $u=Lx$, only $\gamma\rho(M)<1$ can be derived if there exists a unique solution $P\in \mathcal{S}_+^n$ to SLE \eqref{eq5}, while from  \cite[Theorem 1]{kubrusly1985mean}, the admissibility of the control policy requires $\rho(M)<1$.
\end{rem}

The following lemma provides a sufficient condition for testing the admissibility of a control policy.

\begin{lem}\label{lem4}
A control policy $u=Lx$ is admissible if there is a unique solution $P\in \mathcal{S}_+^n$ to SLE \eqref{eq5} and $P<\frac{Q+L^{\top}RL}{1-\gamma}$.
\end{lem}

\begin{proof}
For any $S\in \mathcal{S}_+^n$, define an operator $\mathcal{V}$:
 $\mathcal{S}_+^n\rightarrow \mathcal{S}^n$ by
\begin{equation}
\nonumber \mathcal{V}(S)=(A+BL)^{\top}S(A+BL)+(C+DL)^{\top}S(C+DL).
\end{equation}
If there is a unique solution $P\in \mathcal{S}_+^n$ to SLE \eqref{eq5},
 one has
\begin{equation}\label{lem4_1}
P=\mathcal{V}(P)+\frac{L^{\top}RL+Q-(1-\gamma)P}{\gamma}.
\end{equation}
Considering $P<\frac{Q+L^{\top}RL}{1-\gamma}$, one has that $(L^{\top}RL+Q-(1-\gamma)P)\in\mathcal{S}_+^n$.
Let $\bar{F}=
\frac{L^{\top}RL+Q-(1-\gamma)P}{\gamma}\in\mathcal{S}_+^n$,
 then equation \eqref{lem4_1} becomes
\begin{equation}
\nonumber P=\mathcal{V}(P) +\bar{F}.
\end{equation}
It follows from \cite[Theorem 1]{kubrusly1985mean} that
 the control policy $u=Lx$ is admissible.
\end{proof}

Based on the definition of $V(x_k)$, one has
\begin{align}
\nonumber V(x_k)={\rm E}\big(c(x_k,u_k)\big)+\gamma {\rm E}
\big(\sum_{i=k+1}^{\infty}\gamma^{i-k-1}c(x_i,u_i)\big),
\end{align}
which yields a Bellman equation for cost function:
\begin{align}\label{eq6}
V(x_k)={\rm E}\big(c(x_k,u_k)\big)+\gamma V(x_{k+1}).
\end{align}
Substituting cost function \eqref{eq4} and $ {\rm E}\big(c(x_k,u_k)\big)={\rm E}(x_k^{\top}Qx_k+u_k^{\top}Ru_k)$ into equation \eqref{eq6},
the Bellman equation in terms of
the cost function kernel matrix $P$ is obtained as
\begin{align}\label{eq7}
\nonumber {\rm E}(x_k^{\top}Px_k)&={\rm E}(x_k^{\top}Qx_k+u_k^{\top}Ru_k)\\
&\quad +\gamma{\rm E}(x_{k+1}^{\top}Px_{k+1})-\gamma {\rm {tr}}(PW).
\end{align}
Define the Hamiltonian
\begin{align}
\nonumber H(x_k,u_k)&={\rm E}(x_k^{\top}Qx_k+u_k^{\top}Ru_k)
+\gamma{\rm E}(x_{k+1}^{\top}Px_{k+1})\\
\nonumber&\quad -{\rm E}(x_k^{\top}Px_k)-\gamma {\rm {tr}}(PW),
\end{align}
or equivalently,
 \begin{align}
\nonumber H(x_k,L)&={\rm E}\big(x_k^{\top}(Q+L^{\top}RL)x_k\big)
+\gamma{\rm E}(x_{k+1}^{\top}Px_{k+1})\\
\nonumber&\quad -{\rm E}(x_k^{\top}Px_k)-\gamma {\rm {tr}}(PW).
\end{align}
The next lemma shows that the stochastic LQR problem can be solved based on a stochastic algebra Riccati
equation (SARE).

\begin{lem}\label{lem5}
Under Assumption \ref{ass1}, the optimal control policy for the stochastic LQR problem is
\begin{align}\label{eq_lem5_0}
u^{\ast}=L^{\ast}x,
\end{align}
where the optimal control gain is computed as
\begin{equation}\label{eq_lem5_1}
L^{\ast}=-(R+\gamma B^{\top}P^{\ast}B+\gamma D^{\top}P^{\ast}D)^{-1}
(\gamma B^{\top}P^{\ast}A+\gamma D^{\top}P^{\ast}C),
\end{equation}
and $P^{\ast}\in \mathcal{S}_+^n$ is the unique solution to the following SARE
\begin{align}\label{eq_lem5_2}
\nonumber P&^{\ast}=Q+\gamma A^{\top}P^{\ast}A+\gamma C^{\top}P^{\ast}C
-(\gamma A^{\top}P^{\ast}B+\gamma C^{\top}P^{\ast}D)\\
        &\times(R+\gamma B^{\top}P^{\ast}B+\gamma D^{\top}P^{\ast}D)^{-1}
        (\gamma B^{\top}P^{\ast}A+\gamma D^{\top}P^{\ast}C).
\end{align}

\end{lem}

\begin{proof}
The first-order necessary condition for optimality \cite{lewis2012reinforcement} is given by
\begin{align}\label{proof_lem5_2}
\nonumber \frac{\partial H(x_k,L)}{\partial L}
&=2(R+\gamma B^{\top}PB+\gamma D^{\top}PD)L{\rm E}(x_kx_k^{\top})\\
\nonumber  &\quad  +2(\gamma B^{\top}PA+\gamma D^{\top}PC){\rm E}(x_kx_k^{\top})\\
 &=0.
\end{align}
Note that $R+\gamma B^{\top}PB+\gamma D^{\top}PD$ is positive definite for any $P\in \mathcal{S}_+^n$ and one has ${\rm E}(x_kx_k^{\top})\in \mathcal{S}_+^n$ from Remark \ref{rem_Exx}.
 Hence, the optimal control gain
$L^{\ast}$ is obtained as \eqref{eq_lem5_1}.

SARE \eqref{eq_lem5_2} can be obtained by substituting \eqref{eq1} and \eqref{eq_lem5_1} into equation \eqref{eq7}.
\end{proof}

\begin{rem}\label{rem4}
The optimal control policy is closely related to the discount factor $\gamma$.
Note that substituting \eqref{eq_lem5_1} into SARE \eqref{eq_lem5_2}, one has that $P^{\ast}$ and $L^{\ast}$ satisfy SLE \eqref{eq5}.
However, from Remark \ref{rem3}, the existence of a unique positive definite solution to \eqref{eq5} cannot guarantee the admissibility
of $L^{\ast}$.
In practice, one can gradually increase $\gamma$ to obtain an admissible optimal control policy according to Lemma \ref{lem4}.
A lower bound $\gamma^{\ast}>\bar{c}$ of the discount factor $\gamma$ can be found from \cite[Corollary 3]{7588063}, where the $\bar{c}$ is obtained by solving the linear matrix inequalities.
\end{rem}

\section{model-based RL to solve stochastic LQR}


In this section, an offline PI (Algorithm 1)
is proposed to solve the stochastic LQR problem.
In Algorithm 1, a set of control gains are evaluated in an offline manner.
Moreover, the system matrices are required
in both Policy Evaluation step and Policy Update step.
The convergence of this algorithm to the optimal admissible control gain is proved in Lemma \ref{lem10}, which is an extension of \cite[Theorem 1]{hewer1971iterative}.

\begin{algorithm}[htb]%
\caption{Offline PI} 
\hspace*{0.02in} {\bf Input:} 
Admissible control gain $L^{(0)}$, discount factor $\gamma$, maximum number of iterations $i_{max}$, convergence tolerance $\varepsilon$\\
\hspace*{0.02in} {\bf Output:} 
The estimated optimal control gain $\hat{L}$
\begin{algorithmic}[1]
\For{$i=0:i_{max}$} 
　　\State \textbf{Policy Evaluation:} 
\begin{align}\label{off_pi_1}
\nonumber P^{(i)}&=\gamma (A+BL^{(i)})^{\top}P^{(i)}(A+BL^{(i)})+\gamma(C+DL^{(i)})^{\top}\\
 &\quad \times P^{(i)}(C+DL^{(i)})+(L^{(i)})^{\top}RL^{(i)}+Q
 \end{align}
　　\State  \textbf{Policy Improvement:}
\begin{align}\label{off_pi_2}
\nonumber L^{(i+1)}&=-(R+\gamma B^{\top}P^{(i)}B+\gamma D^{\top}P^{(i)}D)^{-1}\\
&\quad \times(\gamma B^{\top}P^{(i)}A+\gamma D^{\top}P^{(i)}C)
\end{align}
　　\If{$\|L^{(i+1)}-L^{(i)}\|<\varepsilon$}
        \State Break
　　\EndIf\State\textbf{endif}
\EndFor\State\textbf{endfor}
\State $\hat{L}=L^{(i+1)}$
\end{algorithmic}
\end{algorithm}

\begin{lem}\label{lem10}
Given an initial admissible control gain $L^{(0)}$.
Consider the two sequences $\{P^{(i)}\}_{i=0}^\infty$ and $\{L^{(i)}\}_{i=1}^\infty$
obtained from Algorithm 1. If the discount factor $\gamma$ is chosen properly large (less than 1),
 then, for $i=0,1,2\cdots$, the following properties hold:
\begin{enumerate}
  \item $P^{\ast}\leq P^{(i+1)}\leq P^{(i)}$;
  \item $\underset{i\rightarrow \infty}{\lim}P^{(i)}=P^{\ast}$,
  $\underset{i\rightarrow \infty}{\lim}L^{(i)}=L^{\ast}$, where $P^{\ast}$ is the solution to SARE \eqref{eq_lem5_2} and $L^{\ast}$ is computed in \eqref{eq_lem5_1};
  \item $L^{\ast}$ and $L^{(i)}$ are admissible.
\end{enumerate}
\end{lem}

\begin{proof}
For any $S\in \mathcal{S}_+^n$,
define an operator $\mathcal{V}^{(i)}$: $\mathcal{S}_+^n\rightarrow \mathcal{S}^n$ by
\begin{align}
\nonumber \mathcal{V}^{(i)}(S)&=(A+BL^{(i)})^{\top}S(A+BL^{(i)})\\
\nonumber &\quad +(C+DL^{(i)})^{\top}S(C+DL^{(i)}).
\end{align}
Define $M^{(i)}=(A+BL^{(i)})\otimes (A+BL^{(i)})+(C+DL^{(i)})\otimes (C+DL^{(i)}).$

One can rewrite equation \eqref{off_pi_1} in Algorithm 1 as
\begin{equation}\label{eq29}
 P^{(i)}=\gamma \mathcal{V}^{(i)}(P^{(i)})+(L^{(i)})^{\top}RL^{(i)}+Q
\end{equation}
For $i=0$, one has
\begin{equation}\label{eq30}
P^{(0)}=\gamma \mathcal{V}^{(0)}(P^{(0)})+(L^{(0)})^{\top}RL^{(0)}+Q.
\end{equation}
Because $L^{(0)}$ is admissible,
$\rho(M^{(0)})<1$ is satisfied according to \cite[Theorem 1]{kubrusly1985mean}. Hence, $\gamma\rho(M^{(0)})<1$. Then, there is a unique solution $P^{(0)}\in \mathcal{S}_+^n$ to \eqref{eq30} according to \cite[Theorem 1]{kubrusly1985mean}.
For equation \eqref{off_pi_2}, $L^{(1)}$ is computed as
\begin{align}
\nonumber L^{(1)}&=-(R+\gamma B^{\top}P^{(0)}B+\gamma D^{\top}P^{(0)}D)^{-1}\\
\nonumber &\quad \times(\gamma B^{\top}P^{(0)}A+\gamma D^{\top}P^{(0)}C).
\end{align}
Therefore, it can be verified that
\begin{align}\label{eq30_2}
\nonumber&\gamma\mathcal{V}^{(0)}(P^{(0)})+(L^{(0)})^{\top}RL^{(0)} \\
\nonumber &=\gamma\mathcal{V}^{(1)}(P^{(0)})+(L^{(1)})^{\top}RL^{(1)}+
(L^{(1)}-L^{(0)})^{\top}\\
 &\quad \times(R+\gamma B^{\top}P^{(0)}B+\gamma D^{\top}P^{(0)}D)(L^{(1)}-L^{(0)}).
\end{align}
Using \eqref{eq30} and \eqref{eq30_2}, one can derive a new equation for $P^{(0)}$:
\begin{align}\label{eq31}
 P^{(0)}&=\gamma\mathcal{V}^{(1)}(P^{(0)})+U^{(1,0)},
\end{align}
where
\begin{align}
\nonumber U^{(1,0)}&=(L^{(1)})^{\top}RL^{(1)}+
(L^{(1)}-L^{(0)})^{\top}\\
\nonumber &\quad \times(R+\gamma B^{\top}P^{(0)}B+\gamma D^{\top}P^{(0)}D)(L^{(1)}-L^{(0)})\\
\nonumber &\quad+(L^{(0)})^{\top}RL^{(0)}+Q
\end{align}
and obviously $U^{(1,0)}\in \mathcal{S}_+^n$.
Because $P^{(0)}\in \mathcal{S}_+^n$ is the unique solution to \eqref{eq30}, it is also the unique solution to \eqref{eq31}.
Based on \cite[Theorem 1]{kubrusly1985mean}, one has $\gamma\rho(M^{(1)})<1$ and there is a unique solution $P^{(1)}\in \mathcal{S}_+^n$ to the following equation
\begin{equation}\label{eq32}
 P^{(1)}=\gamma \mathcal{V}^{(1)}(P^{(1)})+(L^{(1)})^{\top}RL^{(1)}+Q.
\end{equation}
From \eqref{eq31} and \eqref{eq32}, one obtains
\begin{align}\label{eq33}
\nonumber P&^{(0)}-P^{(1)}=\gamma\mathcal{V}^{(1)}(P^{(0)}-P^{(1)})+(L^{(1)}-L^{(0)})^{\top}\\
 &\quad \times(R+\gamma B^{\top}P^{(0)}B+\gamma D^{\top}P^{(0)}D)(L^{(1)}-L^{(0)}),
\end{align}
and thus $(P^{(0)}-P^{(1)})\in \mathcal{S}^n$ is the unique solution  to \eqref{eq33}.
Hence, $P^{(0)}\geq P^{(1)}$.
Repeating the operations of \eqref{eq30}--\eqref{eq33}, one has $P^{(i+1)}\leq P^{(i)}$. 


Note that $\{P^{(i)}\}_{i=0}^\infty$ is a monotonic non-increasing sequence and positive definite. Hence, $\underset{i\rightarrow \infty}{\lim} P^{(i)}=P^{\infty}$ exists.
Taking the limit $i\rightarrow \infty$ of \eqref{eq29} yields
\begin{equation}\label{eq34}
P^{\infty}=\gamma \mathcal{V}^{\infty}(P^{\infty})+(L^{\infty})^{\top}RL^{\infty}+Q
\end{equation}
and
\begin{align}\label{eq35}
\nonumber L^{\infty}&=-(R+\gamma B^{\top}P^{\infty}B+\gamma D^{\top}P^{\infty}D)^{-1}\\
 &\quad \times(\gamma B^{\top}P^{\infty}A+\gamma D^{\top}P^{\infty}C).
\end{align}
Substituting \eqref{eq35} into \eqref{eq34}, we have
\begin{align}\label{eq36}
\nonumber P^{\infty}&=Q+\gamma A^{\top}P^{\infty}A+\gamma C^{\top}P^{\infty}C
-(\gamma A^{\top}P^{\infty}B\\
\nonumber  &\quad+\gamma C^{\top}P^{\infty}D)(R+\gamma B^{\top}P^{\infty}B+\gamma D^{\top}P^{\infty}D)^{-1}\\
  &\quad\times
 (\gamma B^{\top}P^{\infty}A+\gamma D^{\top}P^{\infty}C).
\end{align}
From Lemma \ref{lem5}, one knows that $P^{\ast}$ is the unique positive definite solution to SARE \eqref{eq36}. Thus, $P^{\infty}=P^{\ast}$ and $L^{\infty}=L^{\ast}$, which implies $P^{\ast}\leq P^{(i)}$.
The proofs of 1) and 2) are completed.

From Remark \ref{rem4} and \cite{7588063}, one knows the control gains $L^{\ast}$ and $L^{(i)}$ are admissible if the discount factor $\gamma$ is chosen properly large.
\end{proof}

\section{model-free RL to solve stochastic LQR}

To remove the requirement of complete knowledge of the
system matrices, a new model-free learning algorithm is proposed to solve the stochastic LQR problem in this section.

Based on Bellman equation \eqref{eq6},
define a Q-function as
\begin{equation}\label{eq14}
Q(x_k,\eta_k)={\rm E}\big(c(x_k,\eta_k)\big) + \gamma V(x_{k+1}),
\end{equation}
where $\eta_k$ is an arbitrary control input at time $k$ and $u=Lx$ is used to calculate $V(x_{k+1})$ for time $k+1,k+2,\ldots$.

From \eqref{eq6} and \eqref{eq14}, if $\eta_k=u_k$, one has
\begin{equation}\label{eq15}
Q(x_k,u_k)=V(x_k).
\end{equation}

Substituting \eqref{eq4} into \eqref{eq15}, the Q-function becomes
\begin{align}\label{eq57}
\nonumber &Q(x_k,u_k)\\
\nonumber&=\gamma \big({\rm E}(x_{k+1}^{\top}Px_{k+1})
+\frac{\gamma}{1-\gamma}{\rm {tr}}(PW)\big)
+ {\rm E}\big(c(x_k,u_k)\big)\\
\nonumber &=\gamma {\rm E}\Big(\big(Ax_k+Bu_k+(Cx_k+Du_k)d_k+w_k\big)^{\top}P\\
\nonumber &\quad \times \big(Ax_k+Bu_k+(Cx_k+Du_k)d_k+w_k\big)\\
\nonumber &\quad +\frac{\gamma}{1-\gamma}{\rm {tr}}(PW)\Big)+{\rm E}(x_k^{\top}Qx_k+u_k^{\top}Ru_k)\\
\nonumber &={\rm E}\big(x_k^{\top}(Q+\gamma A^{\top}PA+\gamma C^{\top}PC)x_k+2\gamma x_k^{\top}(A^{\top}PB\\
\nonumber &\quad +C^{\top}PD)u_k+u_k^{\top}(R+\gamma B^{\top}PB+\gamma D^{\top}PD)u_k\big)\\
 \nonumber&\quad +\frac{\gamma}{1-\gamma}{\rm {tr}}(PW)\\
&={\rm E}\big(\begin{bmatrix}x_k\\
                        u_k
         \end{bmatrix}^{\top}
         H
         \begin{bmatrix}x_k\\
                        u_k
         \end{bmatrix}\big)+\frac{\gamma}{1-\gamma}{\rm {tr}}(PW),
\end{align}
where
\begin{equation}
\nonumber H=\begin{bmatrix}  H_{xx}&H_{xu}\\
                             H_{ux}&H_{uu}
         \end{bmatrix}\in\mathcal{S}_+^{n+m},
\end{equation}
and
\begin{align}
\nonumber H_{xx}&=Q+\gamma A^{\top}PA+\gamma C^{\top}PC\\
\nonumber H_{xu}&=\gamma A^{\top}PB+\gamma C^{\top}PD=H_{ux}^{\top}\\
\nonumber  H_{uu}&=R+\gamma B^{\top}PB+\gamma D^{\top}PD.
\end{align}

Define the optimal Q-function as \cite{watkins1989learning}
\begin{equation}
\nonumber Q^{\ast}(x_k,u_k)={\rm E}\big(c(x_k,u_k)\big)+\gamma V^{\ast}(x_{k+1}).
\end{equation}
By solving
$\frac{\partial Q^{\ast}(x_k,u_k)}{\partial u_k}=0$, the optimal control gain is obtained as
\begin{equation}\label{eq59}
L^{\ast}=-(H^{\ast}_{uu})^{-1}H^{\ast}_{ux},
\end{equation}
where $H^{\ast}_{uu}=R+\gamma B^{\top}P^{\ast}B+\gamma D^{\top}P^{\ast}D$,
$H_{ux}=\gamma B^{\top}P^{\ast}A=H_{xu}^{\top}$, and $P^{\ast}$ satisfies SARE \eqref{eq_lem5_2}.

From \eqref{eq4}, \eqref{eq15}, \eqref{eq57} and the positive definiteness  of ${\rm E}(x_kx_k^{\top})$ in Remark 1, one has
\begin{equation}\label{eq60}
P=      \begin{bmatrix}I\\
                        L
        \end{bmatrix}^{\top}
        H
        \begin{bmatrix}I\\
                        L
        \end{bmatrix}.
\end{equation}
Substituting \eqref{eq60} into \eqref{eq57},
the Q-function can be computed as
\begin{align}\label{eq58_2}
\nonumber &Q(x_k,u_k)\\
&={\rm E}\big(\begin{bmatrix}x_k\\
                        u_k
         \end{bmatrix}^{\top}
                            H
         \begin{bmatrix}  x_k\\
                          u_k
         \end{bmatrix}\big)
         +\frac{\gamma}{1-\gamma}{\rm {tr}}\big(
         H
         \begin{bmatrix}I\\
                   L
    \end{bmatrix} W
    \begin{bmatrix}I\\
                   L
    \end{bmatrix}^{\top}\big).
\end{align}

Based on \eqref{eq14} and \eqref{eq15},  one obtains the Bellman function for Q-function:
\begin{equation}\label{eq16}
 Q(x_k,u_k)={\rm E}\big(c(x_k,u_k)\big) + \gamma Q(x_{k+1},u_{k+1}).
\end{equation}

Using \eqref{eq58_2}, \eqref{eq16}, one has
\begin{align}\label{eq16_2}
\nonumber &{\rm E}\big(\begin{bmatrix}x_k\\
                        u_k
         \end{bmatrix}^{\top}
H
         \begin{bmatrix}  x_k\\
                          u_k
         \end{bmatrix}\big)\\
\nonumber &={\rm E}\big(c(x_k,u_k)\big)
         +\gamma {\rm E}\big(
         \begin{bmatrix}x_{k+1}\\
                        u_{k+1}
         \end{bmatrix}^{\top}
         H
         \begin{bmatrix}  x_{k+1}\\
                          u_{k+1}
         \end{bmatrix} \big)        \\
 &\quad-\gamma {\rm {tr}}\big( H\begin{bmatrix}I\\
                   L
    \end{bmatrix} W
    \begin{bmatrix}I\\
                   L
    \end{bmatrix}^{\top}\big).
\end{align}

Now, we propose a new model-free learning algorithm in Algorithm 2 according to \eqref{eq59} and \eqref{eq16_2}.

\begin{algorithm}[htb]\label{alg3}
\caption{Online Model-free RL} 
\hspace*{0.02in} {\bf Input:} 
Admissible control gain $L^{(0)}$, initial state covariance matrix $X_0$, additive noise covariance matrix $W$, discount factor $\gamma$, maximum number of iterations $i_{max}$, convergence tolerance $\varepsilon$\\
\hspace*{0.02in} {\bf Output:} 
The estimated optimal control gain $\hat{L}$
\begin{algorithmic}[1]
\For{$i=0:i_{max}$} 
　　\State \textbf{Policy Evaluation:}
\begin{align}\label{on_pi_q_1}
\nonumber &{\rm E}\big(\begin{bmatrix}x_k\\
                        u^{(i)}_k
         \end{bmatrix}^{\top}
H^{(i)}
         \begin{bmatrix}  x_k\\
                          u^{(i)}_k
         \end{bmatrix}\big)\\
\nonumber&={\rm E} \big(c(x_k,u^{(i)}_k)\big)
+\gamma {\rm E}\big(
         \begin{bmatrix}x_{k+1}\\
                        u^{(i)}_{k+1}
         \end{bmatrix}^{\top}
         H^{(i)}
         \begin{bmatrix}  x_{k+1}\\
                          u^{(i)}_{k+1}
         \end{bmatrix}\big)\\
&\quad  -\gamma {\rm {tr}}\big( H^{(i)}\begin{bmatrix}I\\
                   L^{(i)}
    \end{bmatrix} W
    \begin{bmatrix}I\\
                   L^{(i)}
    \end{bmatrix}^{\top}\big)
\end{align}
　　\State  \textbf{Policy Improvement:}
\begin{align}\label{on_pi_q_2}
L^{(i+1)}=-(H^{(i)}_{uu})^{-1}H^{(i)}_{ux}
\end{align}
　　\If{$\|L^{(i+1)}-L^{(i)}\|<\varepsilon$}
        \State Break
　　\EndIf\State\textbf{endif}
\EndFor\State\textbf{endfor}
\State $\hat{L}=L^{(i+1)}$
\end{algorithmic}
\end{algorithm}


\begin{rem}\label{rem5}
Compared with
Algorithm 1, Algorithm 2 evaluates the iterative matrix $H^{(i)}$ in an online manner using data acquired along the system trajectories. Moreover, Policy Improvement step in Algorithm 2 is carried out in terms of the learned kernel
matrix $H^{(i)}$ without resorting to the system matrices.
\end{rem}


\begin{lem}\label{lem9}
Algorithm 2 is equivalent to Algorithm 1 in the sense that  equations \eqref{off_pi_1} and \eqref{on_pi_q_1} are equivalent, and equations \eqref{off_pi_2} and \eqref{on_pi_q_2} are equivalent.
\end{lem}

\begin{proof}
Substituting $c(x_k,u^{(i)}_k)=x_k^{\top}Qx_k+(u^{(i)}_k)^{\top}Ru^{(i)}_k$
and $u_k^{(i)}=L^{(i)}x_k$ into equation \eqref{on_pi_q_1}, one obtains
\begin{align}
\nonumber &{\rm E}
\big(x_k^{\top} \begin{bmatrix}I\\
                        L^{(i)}
         \end{bmatrix}^{\top}
H^{(i)}
         \begin{bmatrix}  I\\
                          L^{(i)}
         \end{bmatrix} x_k\big)
         ={\rm E} (x_k^{\top}Qx_k+(u^{(i)}_k)^{\top}Ru^{(i)}_k)\\
\nonumber &\quad +\gamma {\rm E}\big(
x_{k+1}^{\top} \begin{bmatrix}I\\
                        L^{(i)}
         \end{bmatrix}^{\top}
H^{(i)}
         \begin{bmatrix}  I\\
                          L^{(i)}
         \end{bmatrix} x_{k+1}\big)  \\
\nonumber&\quad-\gamma {\rm {tr}}\big( H^{(i)}\begin{bmatrix}I\\
                   L^{(i)}
    \end{bmatrix} W
    \begin{bmatrix}I\\
                   L^{(i)}
    \end{bmatrix}^{\top}\big).
\end{align}
Based on \eqref{eq60}, the above equation can be rewritten as
\begin{align}\label{lem9_2}
\nonumber &{\rm E}(x_k^{\top}P^{(i)}x_k)
={\rm E} (x_k^{\top}Qx_k+(u^{(i)}_k)^{\top}Ru^{(i)}_k)\\
&\quad +\gamma {\rm E}(x_{k+1}^{\top}P^{(i)}x_{k+1})
-\gamma {\rm {tr}}( P^{(i)}W).
\end{align}
Applying system \eqref{eq1} and $u_k^{(i)}=L^{(i)}x_k$ to \eqref{lem9_2}, one has
\begin{align}
\nonumber &{\rm E}(x_k^{\top}P^{(i)}x_k)\\
\nonumber &={\rm E}(x_k^{\top}Qx_k+(u^{(i)}_k)^{\top}Ru^{(i)}_k)
+\gamma{\rm E}(x_{k+1}^{\top}P^{(i)}x_{k+1})\\
\nonumber&\quad
 -\gamma {\rm {tr}}(P^{(i)}W)\\
\nonumber&={\rm E} \big(x_k^{\top}(Q+(L^{(i)})^{\top}RL^{(i)})x_k\big)
+\gamma {\rm E}\Big(\big((A+BL^{(i)})x_k\\
\nonumber&\quad +(C+DL^{(i)})x_kd_k+w_k\big)^{\top}
            P^{(i)}
            \big((A+BL^{(i)})x_k\\
\nonumber&\quad +(C+DL^{(i)})x_kd_k+w_k\big)\Big)
-\gamma {\rm {tr}}( P^{(i)}W)\\
\nonumber &={\rm E}\big(x_k^{\top}(\gamma (A+BL^{(i)})^{\top}P^{(i)}(A+BL^{(i)})+\gamma(C+DL^{(i)})^{\top}\\
 &\quad \times P^{(i)}(C+DL^{(i)})+(L^{(i)})^{\top}RL^{(i)}+Q)x_k\big).
\nonumber\end{align}
Due to the positive definiteness  of ${\rm E}(x_kx_k^{\top})$ from Remark \ref{rem_Exx}, we  conclude that Policy Evaluation step
in Algorithm 2 is equivalent to Policy Evaluation step in Algorithm 1.
Moreover, from equation \eqref{eq57}, the equation \eqref{on_pi_q_2} in Algorithm 2 is equivalent to the equation \eqref{off_pi_2} in Algorithm 1. This completes the proof.
\end{proof}

\begin{thm}\label{thm1}
Consider the two sequences $ \{L^{(i)}\}_{i=1}^\infty $ and $ \{H^{(i)}\}_{i=0}^\infty $ obtained in Algorithm 3, then
$\underset{i\rightarrow\infty}{\lim} L^{(i)}=L^{\ast}$ and $\underset{i\rightarrow\infty}{\lim} H^{(i)}=H^{\ast}$,
which means that
$\underset{i\rightarrow\infty}{\lim} u^{(i)}=u^{\ast}$ and $\underset{i\rightarrow\infty}{\lim} Q^{(i)}=Q^{\ast}$.
\end{thm}

\begin{proof}
Combing Lemma \ref{lem9} and Lemma \ref{lem10} leads to  Theorem \ref{thm1}.
\end{proof}

\section{implementation of online model-free RL algorithm}
The mathematical expectations  are difficult to implement when the system matrices are unknown.
In this section, for Algorithm 2, a numerical average is adopted to approximate expectation and batch least squares (BLS) \cite{lagoudakis2003least}
is employed to estimate the kernel matrix $H^{(i)}$ of the Q-function.
The implementation of Algorithm 2 is given in Algorithm 3.




By vectorization, Equation \eqref{on_pi_q_1} in Algorithm 2 can be rewritten as
\begin{align}
\nonumber&{\rm E}
\big(\phi^{\top}(z^{(i)}_k)\big){\rm {vecs}}(H^{(i)})\\
\nonumber&={\rm E}\big(c(z^{(i)}_k)\big)+
\gamma {\rm E}\big(\phi^{\top}(z^{(i)}_{k+1})\big){\rm {vecs}}(H^{(i)})\\
\nonumber&-\gamma {\rm {vech}}(\varsigma^{(i)}){\rm {vecs}}(H^{(i)}),
\end{align}
where
\begin{align}
 \nonumber  z^{(i)}_k=[x_k^{\top}~(u^{(i)}_k)^{\top}]^{\top}\in \mathbb{R}^{n+m=p} &,~ \phi(z^{(i)}_k)={\rm {vech}}\big(z^{(i)}_k(z^{(i)}_k)^{\top}\big)
\end{align}
 and
\begin{equation}
\nonumber \varsigma^{(i)}=\begin{bmatrix}I\\
                   L^{(i)}
    \end{bmatrix} W
    \begin{bmatrix}I\\
                   L^{(i)}
    \end{bmatrix}^{\top}.
\end{equation}

The kernel matrix $H^{(i)}$ is estimated from data generated under the control policy $u_k=L^{(i)}x_k$ for $N$ time steps.
The BLS estimator of $H^{(i)}$ is given by \cite{lagoudakis2003least}:
\begin{align}\label{eq73}
\nonumber {\rm {vecs}}(H^{(i)})&=\big(({\rm E}\Phi^{(i)})^{\top}({\rm E}\Phi^{(i)}+\gamma \Gamma^{(i)}-\gamma{\rm E} \Psi^{(i)})\big)^{-1}\\
&\quad \times({\rm E}\Phi^{(i)})^{\top}{\rm E}\Upsilon^{(i)},
\end{align}
where $\Phi^{(i)}\in \mathbb{R}^{N\times \frac{p(p+1)}{2}}$, $\Psi^{(i)}\in \mathbb{R}^{N\times \frac{p(p+1)}{2}}$ and  $\Upsilon^{(i)}\in \mathbb{R}^{N}$ are the data matrices defined by
\begin{align}\label{eq74}
\nonumber &\Phi^{(i)}=\begin{bmatrix}
\phi(z^{(i)}_1)&\phi(z^{(i)}_2)&\cdots&\phi(z^{(i)}_N)
\end{bmatrix}^{\top}\\
\nonumber &\Psi^{(i)}=\begin{bmatrix}
\phi(z^{(i)}_2)&\phi(z^{(i)}_3)&\cdots&\phi(z^{(i)}_{N+1})
\end{bmatrix}^{\top}\\
 &\Upsilon^{(i)}=\begin{bmatrix}
c(z^{(i)}_2)&c(z^{(i)}_3)&\cdots&c(z^{(i)}_{N+1})
\end{bmatrix}^{\top},
\end{align}
and $\Gamma^{(i)}$ is a $N\times \frac{p(p+1)}{2}$ matrix whose rows are vectors
${\rm {vech}}(\varsigma^{(i)})$.

Note that $u^{(i)}_k$ is linearly dependent on $x_k$.
Therefore, the BLS estimate equation \eqref{eq74} is not solvable.
To overcome this problem, a probing noise is added to $u^{(i)}_k$ and enough data are collected to ensure the condition ${\rm rank}(\Phi^{(i)})^{\top}\Phi^{(i)})=\frac{p(p+1)}{2}$ holds \cite{DV2009,kiumarsi2017optimal}.

Algorithm 3 is a practical implementation of Algorithm 2.

\begin{algorithm}[htb]\label{alg4}
\caption{Implementation of Online Model-free RL} 
\hspace*{0.02in} {\bf Input:} 
Admissible control gain $L^{(0)}$, initial state covariance matrix $X_0$, additive noise covariance matrix $W$, discount factor $\gamma$, roll out length $N$, maximum number of iterations $i_{max}$,  positive number $num_{mean}$, convergence tolerance $\varepsilon$\\
\hspace*{0.02in} {\bf Output:} 
The estimated optimal control gain $\hat{L}$
\begin{algorithmic}[1]
\For{$i=0:i_{max}$} 
　　\State  Sample $x_0$ from a Gaussian distribution with zeros
    \State  mean and covariance $X_0$. Let $\Phi=0$, $\Psi=0$,$\Upsilon= 0$.
    \State Let $\Phi^{(i)}=0$, $\Psi^{(i)}=0$,$\Upsilon^{(i)}= 0$
　　\For {$q=1:num_{mean}$}
        \State Apply  $u_k=L^{(i)}x_k$ for $N$ time steps, collect data \State and construct $\Phi^{(i)}$, $\Psi^{(i)}$ and $\Upsilon^{(i)}$ through \eqref{eq74}.
        \State $\Phi=\Phi+\Phi^{(i)}$, $\Psi=\Psi+\Psi^{(i)}$, $\Upsilon=\Upsilon+\Upsilon^{(i)}$
    \EndFor\State\textbf{endfor}
    \State $\Phi=\Phi/num_{mean}$, $\Psi=\Psi/num_{mean}$ \State $\Upsilon=\Upsilon/num_{mean}$
       \State ${\rm {vecs}}(H^{(i)})=\big((\Phi)^{\top}(\Phi+\gamma \Gamma^{(i)}-\gamma \Psi)\big)^{-1}\Phi^{\top}\Upsilon$
        \State $L^{(i+1)}=-(H^{(i)}_{uu})^{-1}H^{(i)}_{ux}$
　　\If{$\|L^{(i+1)}-L^{(i)}\|<\varepsilon$}
        \State Break
　　\EndIf\State\textbf{endif}
\EndFor\State\textbf{endfor}
\State $\hat{L}=L^{(i+1)}$
\end{algorithmic}
\end{algorithm}

\begin{rem}\label{rem7}
At each iteration of Algorithm 3,
the BLS solution \eqref{eq73} is employed to estimate
the Q-function  kernel matrix $H^i$,
where the expectations ${\rm E}\Phi^{(i)}$, ${\rm E}\Psi^{(i)}$ and ${\rm E}\Upsilon^{(i)}$ are
approximated by the numerical averages $\Phi$, $\Psi$ and $\Upsilon$, respectively.
\end{rem}

\begin{rem}\label{rem8}
In \cite{WANG20181}, the system matrices are partly required,
in this paper the knowledge of the system matrices is not required.
Compared with \cite{WANG2016379}, no model neural network is used in this paper.
The authors of \cite{JZP2016noise} and \cite{ZHANG20201} make the assumption that the noises is measurable, here we remove this assumption.
Furthermore, our system \eqref{eq1} is more general,
and the proposed model-free learning algorithm is easier to understand and to implement.
\end{rem}

\section{numerical example}
In this section, a numerical example is presented to evaluate our model-free algorithm.

Consider the following stochastic linear discrete-time system
\begin{align}\label{exam}
\nonumber       &x_{k+1}=\begin{bmatrix}
                        0.8 &1\\
                        1.1&2
                        \end{bmatrix}x_k+
                        \begin{bmatrix}
                        0.2\\1.4
                        \end{bmatrix}u_k
                        + \big(\begin{bmatrix}
                        0.7 &0\\
                        -1 &-0.5
                        \end{bmatrix}x_k\\
             &\qquad~~  +\begin{bmatrix}
                        -1\\
                        0.8
                        \end{bmatrix}u_k\big)d_k
                        +w_k.
\end{align}
Let initial state variance matrix $X_0=I$ and additive noise covariance matrix $W=I$.
The weight matrices and discount factor are selected as $Q=I$, $R=1$ and $\gamma=0.7$, respectively.
The exact solution to SARE \eqref{eq_lem5_2} is
\begin{align}
\nonumber P^{\ast}=  \begin{bmatrix}
                      8.2254 &  8.0704\\
                      8.0704& 10.3873
                     \end{bmatrix}
\end{align}
and the optimal control gain is
\begin{align}
\nonumber L^{\ast}=[-0.9319~-1.5784].
\end{align}
Thus, one can obtain the optimal cost $V^{\ast}(x_0)= 62.0422$ according to equation \eqref{eq4}.

Choose $L^{(0)}=[-1.4~-2.1]$,
 $i_{max}=20$, $num_{mean}=5$, and $\varepsilon=10^{-2}$.
 Algorithm 3 stops after five iterations and returns the estimated optimal control gain $\hat{L}=[-0.9369~ -1.5772]$ and  the estimated optimal cost $\hat{V}(x_0)=62.1118$.

Fig. 1 shows the comparison between Algorithm 1 and Algorithm 3.
We can see that Fig. 1 verifies the equivalence of Algorithm 1 and Algorithm 2.
The control gains obtained using Algorithm 3 are comparable to those of Algorithm 1.
The control gains obtained using Algorithm 3 have some small fluctuations near the control gains generated by Algorithm 1.
The fluctuations may be due to the numerical average replacement used in Algorithm 3.
\begin{figure}[hbt]
  \centering
  \includegraphics[width=.5\textwidth]{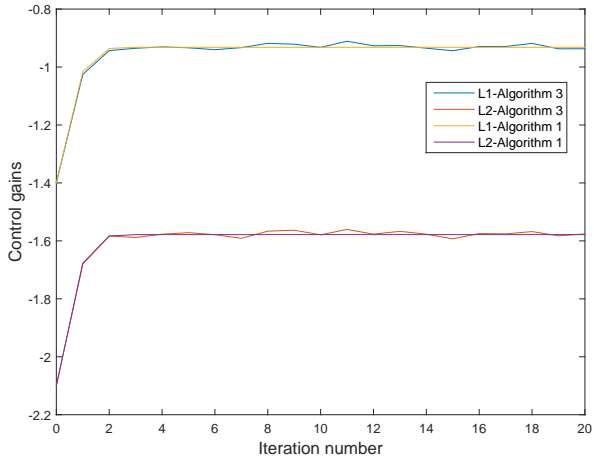}
  \caption{ The control gain curves for Algorithm 1,3 }
\end{figure}


The following PI algorithms are evaluated on system \eqref{exam} with the same initial control gain:
\begin{enumerate}
  \item PI in \cite{bradtke1994adaptive}, where the  kernel matrix $H^{(i)}$ is learned based on recursive least squares.
  \item PI in \cite{JIANG20122699}, where the cost function kernel matrix $P^{(i)}$,  the matrices $H_{uu}^{(i)}$ and $H_{ux}^{(i)}$( not the entire kernel matrix $H^{(i)}$) are  estimated .
  \item MFLQv3 in \cite{pmlr-v89-abbasi-yadkori19a}, where a cost function estimate $\hat{V}_i$ is first computed and then the Q-function estimate $\hat{Q}_i$ is estimated. The iterative control policy is a greedy policy with respect to the average of all previous estimates $\hat{Q}_1,\ldots,\hat{Q}_{i-1}$.
\end{enumerate}
To compare the performance, we run each algorithm 10 times and  90000 time steps is used in a single run.
The average curves of $\|L^{(i)}-L^{\ast} \|$ and $\frac{|V^{(i)}(x_0)-V^\ast(x_0)|}{V^\ast(x_0)}$ are shown in Fig. 2 and Fig. 3, respectively.
We can see, from Fig. 2, that the control gains generated by Algorithm 3 is closer to the optimal control gain.
Fig. 3 shows that Algorithm 3 and MFLQv3 achieve lower relative cost error than PI in \cite{bradtke1994adaptive} and PI in \cite{JIANG20122699}. Moreover, Algorithm 3 needs less iterations to achieve a lower relative cost error.

\begin{figure}[hbt]
  \centering
  \includegraphics[width=.5\textwidth]{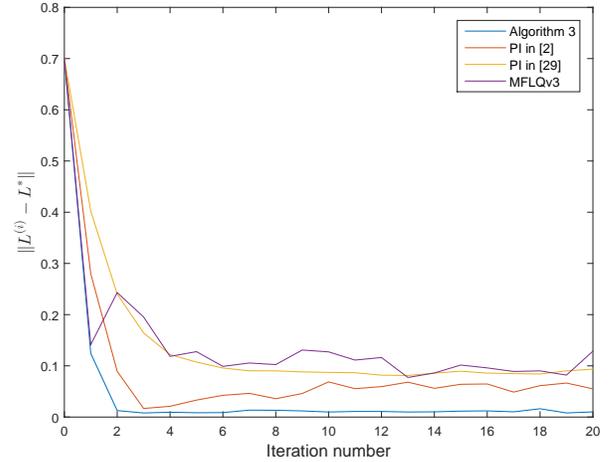}
  \caption{ The control gain distance curves }
\end{figure}

\begin{figure}[hbt]
  \centering
  \includegraphics[width=.5\textwidth]{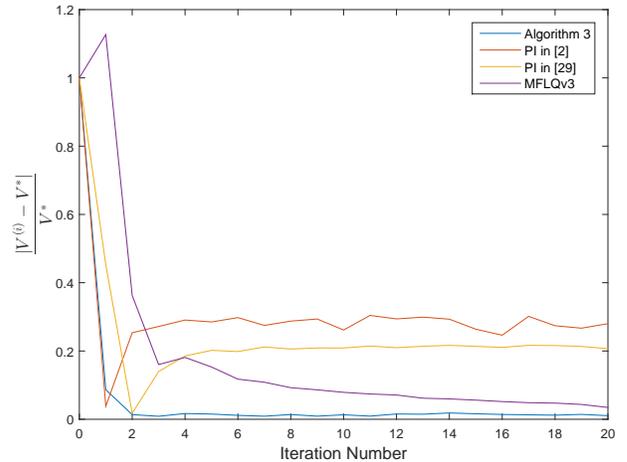}
  \caption{ Relative cost error curves }
\end{figure}



\section{conclusions}
In this paper, the optimal control problem for a
class of discrete-time stochastic systems subject to additive and multiplicative noises has been investigated.
The objective is to find the optimal control policy in the sense of minimizing a discounted cost function and maintaining the asymptotically square stationary property of the system.
To avoid requiring any knowledge of the system matrices, a model-free reinforcement
learning algorithm has been proposed to search for the optimal control policy
 using the data of the system states and control inputs.
The model-free learning algorithm has been implemented through
batch least squares and a numerical average.
The effectiveness of the proposed algorithm
has been illustrated through a numerical example.

\bibliographystyle{IEEEtran}
\bibliography{mybib}


\ifCLASSOPTIONcaptionsoff
  \newpage
\fi

\end{document}